\documentclass[11pt]{CGC2}

\usepackage[english]{babel}
\usepackage{verbatim}
\usepackage{graphicx}
\newcommand{\w}{{\mathrm w}}

\usepackage[english]{babel}
\usepackage{amsmath}
\usepackage{amssymb}
\usepackage{amsthm}
\usepackage[mathcal]{eucal}
\usepackage[latin1]{inputenc}

\usepackage{color}

\begin{document}

%%%%%%%%%%%%%%%%%%%%%%%%%%%%%%%%%%%%%%%%%%%%%%%%%%%%%%%%%%%%%%%%%%%%%
\Logo{}

\begin{frontmatter}

\title{On the Griesmer Bound for Systematic Codes}

{\author{Emanuele Bellini}}
{{\tt (eemanuele.bellini@gmail.com)}}\\
{{Department of Mathematics, University of Trento, Italy.}}

%{\author{Massimiliano Sala}} {\tt{(maxsalacodes@gmail.com)}}\\
%{Department of Mathematics, University of Trento, Italy.}

\runauthor{E.~Bellini}

\begin{abstract}
We generalize the Griesmer bound in the case of systematic codes over a field of size $q$ greater than the distance $d$ of the code.
We also generalize the Griesmer bound in the case of any systematic code of distance 2,3,4 and in the case of binary systematic codes of distance up to 6.
\end{abstract}

\end{frontmatter}
%==================================================================

\section{Introduction}

In this article we want to prove that the Griesmer Bound applies also to some systematic codes.

\section{The Griesmer Bound}
From now on let $q$ be the power of a prime number, and $n,k,d$ three integers such that a $[q,n,k,d]$ systematic code exists.\\
Let us recall the original bound given by Griesmer \cite{gries}.
\begin{theorem}[Griesmer bound]\label{griesmer}
Let $n$ be such that there exists an $[q,n,k]$ linear code with distance at least $d$. Then
$$
 n \geq \sum_{j=0}^{k-1} \left\lceil \frac{d}{q^j} \right\rceil.
$$   
\end{theorem} 

\begin{lemma}\label{lemk1}
  If $k=1$, then for each $q,n,d$ such that a systematic code $C$ exists, then
  $$
  n \geq \sum_{j=0}^{k-1} \left\lceil \frac{d}{q^j} \right\rceil.
  $$  
\end{lemma}
\begin{proof}
  For $k = 1$ we have $\sum_{j=0}^{k-1} \left\lceil \frac{d}{q^j} \right\rceil = d$, and clearly $n \ge d$.
\end{proof}

In the following sections let $C$ be a $[q,n,k,d]$ systematic code, with $k \ge 2$, such that $0 \in C$, and let us indicate a word of $C$ as $c = (\bar{c},\tilde{c})$, where $\bar{c}$ is the systematic part of $c$ and $\tilde{c}$ is the nonsystematic part of $c$.

%%%%%%%%%%%%%%%%%%%%%%%%%%%
%%%% q >= d %%%%%%%%%%%%%%%
%%%%%%%%%%%%%%%%%%%%%%%%%%%
\section{The case $q \ge d$}
\begin{theorem}\label{thmqd}
  If $q \ge d$, for all $k \ge 2$, there exists no $q$\_ary systematic code such that $n < \sum_{i=0}^{k-1}\lceil\frac{d}{q^i}\rceil$.
\end{theorem}
\begin{proof}
  If $q \ge d$ we have that $\lceil\frac{d}{q^i}\rceil = 1$ for all $i \ge 1$, and so:
 \begin{align*}
  \sum_{i=0}^{k-1}\lceil\frac{d}{q^i}\rceil = d + \lceil\frac{d}{q}\rceil + \dots +  \lceil\frac{d}{q^{k-1}}\rceil = d + k - 1
 \end{align*}
 But we also have, by the Singleton bound, that $n \ge d + k - 1$.  
\end{proof}

%%%%%%%%%%%%%%%%%%%%%%%%%%%
%% q >= 2 % d = 1,2,3,4 %%%
%%%%%%%%%%%%%%%%%%%%%%%%%%%
\section{The case $d = 1,2,3,4$}
%%%%%%%%%%%%%%%%%%%%%%%%%%%
%%%%% q >= 2 % d = 1,2 %%%%
%%%%%%%%%%%%%%%%%%%%%%%%%%%
\begin{theorem}\label{thmd12}
  If $d=1,2$, than for all $k \ge 2$, there exists no $q$\_ary systematic code such that $n < \sum_{i=0}^{k-1}\lceil\frac{d}{q^i}\rceil$.
\end{theorem}
\begin{proof}
  If $d=1,2$ we have $q \ge d$ an so we are in the hypothesis of Theorem \ref{thmqd}.
\end{proof}
%%%%%%%%%%%%%%%%%%%%%%%%%%%
%%%% q >= 2 % d = 3,4  %%%%
%%%%%%%%%%%%%%%%%%%%%%%%%%%
\begin{theorem}\label{thmd34}
  If $d=3,4$, then for all $k \ge 2$, there exists no $q$\_ary systematic code such that $n < \sum_{i=0}^{k-1}\lceil\frac{d}{q^i}\rceil$.
\end{theorem}
\begin{proof}
 If $d=3$ and $q \ge 3$ or $d=4$ and $q \ge 4$ then we are in the hypothesis of Theorem \ref{thmqd}.\\

 Otherwise, if $d=3$ and $q=2$ or $d=4$ and $q=2,3$ then we have that $\lceil\frac{d}{q}\rceil = 2$ and $\lceil\frac{d}{q^i}\rceil = 1$ for all $i \ge 2$, and so:
 \begin{align*}
  \sum_{i=0}^{k-1}\lceil\frac{d}{q^i}\rceil = d + \lceil\frac{d}{q}\rceil + \dots + \lceil\frac{d}{q^{k-1}}\rceil = d + k
 \end{align*}
 Suppose by contradiction that $n < d+k$. It is enough to prove the case $n = d+k-1$. Then $n-k = d-1$. Since $0 \in C$, if we consider two different words $c_1 = (\bar{c_1},\tilde{c_1}),c_2 = (\bar{c_2},\tilde{c_2})$ such that $\w(\bar{c_1}) = \w(\bar{c_2}) = 1$, then $\w(\tilde{c_1}) = \w(\tilde{c_2}) = d-1$ and, if $q=2$ then $\tilde{c_1}=\tilde{c_2}=(1,\dots,1)$, and so $\d(c_1,c_2) \le 2$, contradiction.
 There is only one case left, which is the case $q=3$ and $d=4$. In this case, since $k \ge 2$, we have at least $9$ words in $C$. Consider the following four words:
 \begin{align*}
  c_0 = (\bar{c_0},\tilde{c_0}) = & (0 \dots 0 000 , 0 \dots 0) \\
  c_1 = (\bar{c_1},\tilde{c_1}) = & (0 \dots 0 001 , c_{11} c_{12} c_{13}) \\
  c_2 = (\bar{c_2},\tilde{c_2}) = & (0 \dots 0 002 , c_{21} c_{22} c_{23}) \\
  c_3 = (\bar{c_3},\tilde{c_3}) = & (0 \dots 0 010 , c_{31} c_{32} c_{33}) \\
 \end{align*}
 Since the distance between $c_1,c_2,c_3$ from $c_0$ must be greater than $d$, then $\tilde{c_1},\tilde{c_2},\tilde{c_3}$ must have weigth $3$. $c_{11}, c_{12}, c_{13}$ can be any combination of $1$ and $2$, let us suppose $(c_{11}, c_{12}, c_{13}) = (111)$. Then, to have $\d(c_1,c_2) \ge 4$, we must have $(c_{21}, c_{22}, c_{23}) = (222)$. And for the same reason $(c_{31} c_{32} c_{33})$ must differ from $(c_{11} c_{12} c_{13})$ and from
 $(c_{21} c_{22} c_{23})$ in at least two positions at the same time, but this is not possible.
\end{proof}

%%%%%%%%%%%%%%%%%%%%%%%%%%%
%%%% q = 2 % d = 5,6   %%%%
%%%%%%%%%%%%%%%%%%%%%%%%%%%

\section{The case $q=2$ and $d = 5,6$}

\begin{theorem}
  For $k = 2$, there exists no binary systematic code such that $n < \sum_{i=0}^{k-1}\lceil\frac{d}{2^i}\rceil$ for $d=5,6$.
\end{theorem}
\begin{proof}
 If $k=2$ then $ \sum_{i=0}^{1}\lceil\frac{d}{2^i}\rceil = d + \lceil\frac{d}{2}\rceil = d + 3$. Suppose by contradiction that $n < d+3$. It is enough to prove the case $n=d+2$. Consider two different words $c_1 = (\bar{c_1},\tilde{c_1}),c_2 = (\bar{c_2},\tilde{c_2})$ such that $\w(\bar{c_1}) = \w(\bar{c_2}) = 1$,  then $\w(\tilde{c_1}) = \w(\tilde{c_2}) \ge d-1$. Since $n-k = d$, then $\d(\tilde{c_1},\tilde{c_2})\le 2$ and thus $\d(c_1,c_2) \le 4$.
\end{proof}

\begin{theorem}
  For all $k \ge 3$, there exists no binary systematic code such that $n < \sum_{i=0}^{k-1}\lceil\frac{d}{2^i}\rceil$ for $d=5,6$.
\end{theorem}
\begin{proof}
 If $d=5,6$, we have that $\lceil\frac{d}{2}\rceil = 3$, $\lceil\frac{d}{4}\rceil = 2$ and $\lceil\frac{d}{2^i}\rceil = 1$ for all $i \ge 3$, and so:
 \begin{align*}
  \sum_{i=0}^{k-1}\lceil\frac{d}{2^i}\rceil = d + \lceil\frac{d}{2}\rceil + \lceil\frac{d}{4}\rceil + \dots + \lceil\frac{d}{2^{k-1}}\rceil = d + 5 + k - 3 = d + k + 2
 \end{align*}
 Suppose by contradiction that $n < d+k+2$. It is enough to prove the case $n = d+k+1$, so that $n-k = d+1$. Let us consider the following five words:
 \begin{align*}
  c_0 = (\bar{c_0},\tilde{c_0}) = & (0 \dots 0 000 , 0 \dots 0) \\
  c_1 = (\bar{c_1},\tilde{c_1}) = & (0 \dots 0 001 , c_{11} \dots c_{1,d+1}) \\
  c_2 = (\bar{c_2},\tilde{c_2}) = & (0 \dots 0 010 , c_{21} \dots c_{2,d+1}) \\
  c_3 = (\bar{c_3},\tilde{c_3}) = & (0 \dots 0 011 , c_{31} \dots c_{3,d+1}) \\
  c_5 = (\bar{c_5},\tilde{c_5}) = & (0 \dots 0 101 , c_{51} \dots c_{5,d+1}) 
 \end{align*}
 We want to show that there is no way to assign $0$ or $1$ to the $c_{ij}$ to obtain distance $d$ between these five words. Let us do the following considerations:
 \begin{enumerate}
  \item to have $\d(c_1,c_0)=\w(c_1) \ge d$ and $\d(c_2,c_0)=\w(c_2) \ge d$, it must be that $\w(\tilde{c_1}),\w(\tilde{c_2}) \ge d-1$. Clearly it is not possible that $\w(\tilde{c_1}),\w(\tilde{c_2}) \ge d$, otherwise $\d(c_1,c_2) \le 4$. So, wlog, we have only one of the two following cases:
  \begin{enumerate}
   \item either $\w(\tilde{c_1})=d$ and $\w(\tilde{c_2}) = d-1$,
   \item or $\w(\tilde{c_1})=\w(\tilde{c_2})=d-1$.
  \end{enumerate}    
  \item to have $\w(c_3),\w(c_5) \ge d$, it must be that $\w(\tilde{c_3}),\w(\tilde{c_5}) \ge d-2$.
 \end{enumerate}
 Consider the case (1.a). Since $\d(\bar{c_1},\bar{c_3})=1$ and $\w(\tilde{c_1})=d$, the only way to have distance at least $d$ between $c_0,c_1,c_3$ (modulo the permutation of the colums) is to assign the following values to $c_1,c_3$:
 \begin{align*}
  c_0 = (\bar{c_0},\tilde{c_0}) = & (0 \dots 0 000 , 000000) \\
  c_1 = (\bar{c_1},\tilde{c_1}) = & (0 \dots 0 001 , 011111) \\
  c_3 = (\bar{c_3},\tilde{c_3}) = & (0 \dots 0 011 , 111000) \\
  c_2 = (\bar{c_2},\tilde{c_2}) = & (0 \dots 0 010 , c_{21} \dots c_{2,d+1}) \\
  c_5 = (\bar{c_5},\tilde{c_5}) = & (0 \dots 0 101 , c_{51} \dots c_{5,d+1}) 
 \end{align*}  
 in case $d=5$ and:
  \begin{align*}
  c_0 = (\bar{c_0},\tilde{c_0}) = & (0 \dots 0 000 , 0000000) \\
  c_1 = (\bar{c_1},\tilde{c_1}) = & (0 \dots 0 001 , 0111111) \\
  c_3 = (\bar{c_3},\tilde{c_3}) = & (0 \dots 0 011 , 1111000) \\
  c_2 = (\bar{c_2},\tilde{c_2}) = & (0 \dots 0 010 , c_{21} \dots c_{2,d+1}) \\
  c_5 = (\bar{c_5},\tilde{c_5}) = & (0 \dots 0 101 , c_{51} \dots c_{5,d+1}) 
 \end{align*}  
 in the case $d=6$.\\
 This allows to have $\d(\tilde{c_1},\tilde{c_3})=d-1$, which is the only we can reach in our conditions. \\
 Now consider $c_2$. $\tilde{c_1}$ has only a zero and $d$ ones, and $\tilde{c_2}$ has $d-1$ ones, which are either in the same postitions of the ones in $\tilde{c_1}$ (this case is impossible because otherwise $\d(\tilde{c_1},\tilde{c_2})=1\implies \d(c_1,c_2)=3$) or $c_{21}=1$. Since $\w(c_2)=d$, there remain $d$ bits to be filled in $\tilde{c_2}$, and $d-2$ of this bit must be ones and the other $2$ zeros. Since $c_{31} = c_{21}=1$, to have $\d(\tilde{c_3},\tilde{c_2}) \ge d-1$, at least $d-1$ of the $d$ rightmost bits must differ. Thus we have the following situation in case $d=5$:
 \begin{align*}
  c_0 = (\bar{c_0},\tilde{c_0}) = & (0 \dots 0 000 , 000000) \\
  c_1 = (\bar{c_1},\tilde{c_1}) = & (0 \dots 0 001 , 011111) \\
  c_3 = (\bar{c_3},\tilde{c_3}) = & (0 \dots 0 011 , 111000) \\
  c_2 = (\bar{c_2},\tilde{c_2}) = & (0 \dots 0 010 , 100111) \\
  c_5 = (\bar{c_5},\tilde{c_5}) = & (0 \dots 0 101 , c_{51} \dots c_{56}) 
 \end{align*}  
and in the following situation in case $d=6$:
 \begin{align*}
  c_0 = (\bar{c_0},\tilde{c_0}) = & (0 \dots 0 000 , 0000000) \\
  c_1 = (\bar{c_1},\tilde{c_1}) = & (0 \dots 0 001 , 0111111) \\
  c_3 = (\bar{c_3},\tilde{c_3}) = & (0 \dots 0 011 , 1111000) \\
  c_2 = (\bar{c_2},\tilde{c_2}) = & (0 \dots 0 010 , 1001111) \\
  c_5 = (\bar{c_5},\tilde{c_5}) = & (0 \dots 0 101 , c_{51} \dots c_{57}) 
 \end{align*}   
 Now, $\tilde{c_5}$ must be such that $\w(\tilde{c_5})\ge d-2$ and $\d(\tilde{c_5},\tilde{c_1})\ge d-1$. Thus in $\tilde{c_5}$ there must be at most $d+1-(d-2)=3$ zero components, which is a contradiction because, in the case $d=5$, if:
  \begin{align*}
  c_0 = (\bar{c_0},\tilde{c_0}) = & (0 \dots 0 000 , 000000) \\
  c_1 = (\bar{c_1},\tilde{c_1}) = & (0 \dots 0 001 , 011111) \\
  c_3 = (\bar{c_3},\tilde{c_3}) = & (0 \dots 0 011 , 111000) \\
  c_2 = (\bar{c_2},\tilde{c_2}) = & (0 \dots 0 010 , 100111) \\
  c_5 = (\bar{c_5},\tilde{c_5}) = & (0 \dots 0 101 , 100011) 
 \end{align*}  
 or, in the case $d=6$, if:
  \begin{align*}
  c_0 = (\bar{c_0},\tilde{c_0}) = & (0 \dots 0 000 , 0000000) \\
  c_1 = (\bar{c_1},\tilde{c_1}) = & (0 \dots 0 001 , 0111111) \\
  c_3 = (\bar{c_3},\tilde{c_3}) = & (0 \dots 0 011 , 1111000) \\
  c_2 = (\bar{c_2},\tilde{c_2}) = & (0 \dots 0 010 , 1001111) \\
  c_5 = (\bar{c_5},\tilde{c_5}) = & (0 \dots 0 101 , 1000111) 
 \end{align*}  
 then $\d(c_2,c_5) = 4$.\\
 Let us try now with case (1.b), so that we know $\w(\tilde{c_1})=\w(\tilde{c_2})=d-1$. We also have that $\d(\tilde{c_1},\tilde{c_2})\ge d-2$, and at the same time $\d(\tilde{c_1},\tilde{c_2})$ can only be $0,2,4$, since there are only two zeros components both in $\tilde{c_1}$ and in $\tilde{c_2}$, and $c_1$ and $c_2$ have the same parity. Since $d-2 > 2$, then $\d(\tilde{c_1},\tilde{c_2})$ must be $4$ and the only choice (modulo permutation of the columns) for $\tilde{c_1},\tilde{c_2}$ is:
 \begin{align*}
  c_0 = (\bar{c_0},\tilde{c_0}) = & (0 \dots 0 000 , 000000 | 0) \\
  c_1 = (\bar{c_1},\tilde{c_1}) = & (0 \dots 0 001 , 001111 | 1) \\
  c_2 = (\bar{c_2},\tilde{c_2}) = & (0 \dots 0 010 , 111100 | 1) \\
  c_4 = (\bar{c_4},\tilde{c_4}) = & (0 \dots 0 100 , c_{41} \dots c_{4,d+1}) \\ 
  c_3 = (\bar{c_3},\tilde{c_3}) = & (0 \dots 0 011 , c_{31} \dots c_{3,d+1}) 
 \end{align*} 
 where the rightmost component exists only in the case $d=6$.\\
 Now consider $c_4$, which must be such that $\w(\tilde{c_4}) = d-1$ (it can not be $d$ or $d+1$, otherwise we would be in a similar case to (1.a) ), so that it has two zero components which, using a reasoning similar to that for $\tilde{c_1}$ and $\tilde{c_2}$, to have $\d(\tilde{c_1},\tilde{c_4})=\d(\tilde{c_4},\tilde{c_2})=4$, must be positioned as follows:
  \begin{align*}
  c_0 = (\bar{c_0},\tilde{c_0}) = & (0 \dots 0 000 , 000000 | 0) \\
  c_1 = (\bar{c_1},\tilde{c_1}) = & (0 \dots 0 001 , 001111 | 1) \\
  c_2 = (\bar{c_2},\tilde{c_2}) = & (0 \dots 0 010 , 111100 | 1) \\
  c_4 = (\bar{c_4},\tilde{c_4}) = & (0 \dots 0 100 , 110011 | 1) \\
  c_3 = (\bar{c_3},\tilde{c_3}) = & (0 \dots 0 011 , c_{31} \dots c_{3,d+1}) 
 \end{align*} 
 Now, $c_3$ is such that $\w(\tilde{c_3})\ge d-2$.\\ 
 $\w(\tilde{c_3})=d+1$ or $\w(\tilde{c_3})=d$ is not possible, otherwise we would have $\d(c_3,c_1) \le 4$.\\ 
 $\w(\tilde{c_3})=d-1$ is not possible in the case $d=6$, because, having only two zeros the value $\d(c_3,c_1)$ can be at most 5.\\
 In the case $d=5$, if $\w(\tilde{c_3})=d-1$, to have $\d(c_3,c_1) \ge 5$ the two leftmost component of $\tilde{c_3}$ must be the same as the two leftmost component of $\tilde{c_2}$ and of $\tilde{c_4}$, which are ones, giving either $\d(c_3,c_2)=5$ and $\d(c_3,c_4)=3$, or $\d(c_3,c_2) < 5$, which is a contradiction.\\
 It remains to prove that $\w(\tilde{c_3}) \ne d-2$. In this case, again, to have $\d(c_3,c_1) \ge d$ the two leftmost component of $\tilde{c_3}$ must be the same as the two leftmost component of $\tilde{c_2}$ and of $\tilde{c_4}$, which are ones, obtaining actually $\d(c_3,c_1) = 6$. In the remaining components of $\tilde{c_3}$ there must be three zeros. If these three zeros are in the same positions where the leftmost ones of $\tilde{c_2}$ are, then $\d(c_3,c_2) = d$ and $\d(c_3,c_4) \le 4$. Otherwise $\d(c_3,c_2) < d$.\\
 This completes our proof.
 \end{proof}

\end{document}